\documentclass{llncs}

\usepackage{amssymb}
\usepackage{algorithm}
\usepackage{algpseudocode}
\usepackage{tikz}
\usepackage{pgfplots}
\usepackage{comment}

\newcommand{\union}{\cup}
\newcommand{\Union}{\bigcup}

\newcommand{\sop}{[}
\newcommand{\scl}{]}
\newcommand{\sel}[2]{#1 \backslash #2}
\newcommand{\unsubst}[2]{\sop \sel{#1}{#2} \scl}

\newcommand{\impl}{\supset}
\newcommand{\Land}{\bigwedge}
\newcommand{\Lor}{\bigvee}

\usepackage{proof}
\newcommand{\LK}{\ensuremath{\mathbf{LK}}}	
\newcommand{\LKeq}{\ensuremath{\LK_=}}  
\newcommand{\seq}{\rightarrow}	
\newcommand{\com}[1]{|{#1}|}	
\newcommand{\comq}[1]{|{#1}|_q} 

\newcommand{\contrstar}{\mathrm{c}^*}

\newcommand{\cut}{\mathrm{cut}}

\newcommand{\foralll}{\forall_\mathrm{l}}
\newcommand{\forallr}{\forall_\mathrm{r}}

\newcommand{\ElI}{{\rm El1}}
\newcommand{\ElII}{{\rm El2}}
\newcommand{\ErI}{{\rm Er1}}
\newcommand{\ErII}{{\rm Er2}}

\newcommand{\Fcal}{{\cal F}}

\newcommand{\IN}{\subseteq}

\newcommand{\res}{\ensuremath{\mathrm{res}}}

\newcommand{\CI}{{\rm CI}}

\newcommand{\SF}[1]{\ensuremath{\mathrm{SF}_{#1}}}    

\newcommand{\formulas}{\mathfrak{F}}
\newcommand{\solutions}{\mathfrak{S}}
\newcommand{\forgetOp}{\ensuremath{\mathcal{F}}} 
\newcommand{\para}{\ensuremath{\mathrm{para}}} 
\newcommand{\extractTS}{\ensuremath{\mathrm{extractTermSet}}} 
\newcommand{\getMinimalDecomposition}{\ensuremath{\mathrm{getMinimalDecomposition}}} 
\newcommand{\getCanonicalSolution}{\ensuremath{\mathrm{getCanonicalSolution}}} 
\newcommand{\improveSolution}{\ensuremath{\mathrm{improveSolution}}} 
\newcommand{\constructProof}{\ensuremath{\mathrm{constructProof}}}

\newcommand{\DeltaG}{$\Delta_G$}

\newcommand{\mt}[1]{\textnormal{#1}}

\newcommand{\colvec}[1]{\left(\begin{array}{c}#1\end{array}\right)}

\title{Introducing Quantified Cuts\\in Logic with Equality
}
\author{
  Stefan Hetzl\inst{1} \and 
  Alexander Leitsch\inst{2} \and 
  Giselle Reis\inst{2} \and\\
  Janos Tapolczai\inst{1} \and
  Daniel Weller\inst{1}
}

\institute{
  Institut f\"ur Diskrete Mathematik und Geometrie, Technische Universit\"at Wien
  \and
  Institut f\"ur Computersprachen, Technische Universit\"at Wien
}

\begin{document}

\maketitle

\begin{abstract}
Cut-introduction is a technique for structuring and compressing formal proofs.
In this paper we generalize our cut-introduction method for the introduction of
quantified lemmas of the form $\forall x.A$ (for quantifier-free $A$) to a
method generating lemmas of the form $\forall x_1 \ldots \forall x_n.A$.
Moreover, we extend the original method to predicate logic with equality. The
new method was implemented and applied to the TSTP proof database. It is shown
that the extension of the method to handle equality and quantifier-blocks leads
to a substantial improvement of the old algorithm.
\end{abstract}

%
%
%

\section{Introduction}
Computer-generated proofs are typically analytic, i.e., they only contain logical material that also appears in the statement of the theorem. This is due to the fact
that analytic proof systems have a considerably smaller search space which makes proof-search practically feasible. In the case of sequent calculus, proof-search
procedures typically work on the cut-free fragment. But also resolution is essentially analytic as resolution proofs satisfy the subformula property of first-order
logic. One interesting property of non-analytic proofs is their considerably smaller length. The exact difference depends on the logic (or theory) under
consideration, but it is typically enormous. In (classical and intuitionistic) first-order logic there are proofs with cut of length $n$ whose theorems have only
cut-free proofs of length $2_n$ (where $2_0 = 1$ and $2_{n+1}=2^{2_n}$) (see~\cite{Statman79Lower} and~\cite{Orevkov79Lower}). The length of a proof plays an
important role in many situations such as human readability, space requirements and time requirements for proof checking. For most of these situations general-purpose
data compression methods cannot be used as the compressed representation is not a proof anymore. It is therefore of high practical interest to develop methods of
proof transformation which produce non-analytic and hence potentially much shorter proofs.

Work on cut-introduction can be found at a number of different places in the literature.
Closest to our work are other approaches which aim to abbreviate or structure a
given input proof. In~\cite{WoltzenlogelPaleo10Atomic} an
algorithm for the introduction of atomic cuts that is capable of exponential
proof compression is presented.
The method~\cite{Finger07Equal} for propositional logic is
shown to never increase the size of proofs more than polynomially.
Another approach to the compression of first-order proofs by
introduction of definitions for abbreviating terms is~\cite{Vyskocil10Automated}.
There is a large body of work on the generation of non-analytic formulas carried out
by numerous researchers in various communities. Methods for lemma generation are of crucial importance
in inductive theorem proving which frequently requires generalization~\cite{Bundy01Automation},
see e.g.~\cite{Ireland96Productive} for a method in the context of rippling~\cite{Bundy05Rippling}
which is based on failed proof attempts.
In automated theory formation~\cite{Colton01Automated,Colton02Automated}, an eager
approach to lemma generation is adopted. This work has, for example, led to
automated classification results 
of isomorphism classes~\cite{Sorge08Classification} and isotopy classes~\cite{Sorge08Automatic}
in finite algebra. See also~\cite{Johansson11Conjecture} for an approach to inductive theory
formation.

Methods of {\em algorithmic cut-introduction}, based on the inversion of
Gentzen's cut-elimination method, have been defined in~\cite{Hetzl12Towards}
and~\cite{HetzlXXAlgorithmic}. The method  in~\cite{Hetzl12Towards} works on a
cut-free {\LK}-proof $\varphi$ of a prenex skolemized end-sequent $S$ and
consists of the following steps: (1) extraction of a set of terms $T$ from
$\varphi$, (2) computation of a compressed representation of $T$,
(3) construction of the cut formula, (4) improvement of the solution by
computation smaller cut-formulas, and (5) construction of an {\LK}-proof with the
universal cut formula obtained in (4) and instantiation of the quantifiers with
the terms obtained in (2). It has been shown in~\cite{Hetzl12Towards}
that the method is capable of compressing cut-free proofs quadratically. The
paper~\cite{HetzlXXAlgorithmic} generalized the method to the introduction of
arbitrarily many universal cut formulas, where the steps defined above are
roughly the same, though the improvement of the solution (step 4) and the final
construction of the proof with cuts (step 5) are much more difficult. The method
of introducing arbitrarily many universal cuts in~\cite{HetzlXXAlgorithmic} leads
even to an exponential compression of proof length. Still the methods described
above were mainly designed for a theoretical analysis of the cut-introduction
problem rather than for practical applications. In particular, they lacked
efficient handling of equality (as they were defined for predicate logic without
equality) and the introduction of several universal quantifiers in cut formulas
(all cut formulas constructed in~\cite{HetzlXXAlgorithmic} are of the form
$\forall x.A$ for a single variable $x$ and a quantifier-free formula $A$).

In this paper we generalize our cut-introduction method to predicate logic with
equality and to the construction of a (single) quantified cut containing blocks
of universal quantifiers. The efficient compression of the terms (step 2) and
the improvement of the solution (step 4) require new and non-trivial techniques.
Moreover, we applied the new method in large-scale experiments to proofs
generated by prover9 on the TPTP library. This empirical evaluation demonstrates
the feasibility of our method on realistic examples.

\section{Proofs and Herbrand Sequents}

Throughout this paper we consider predicate logic with equality. For practical
reasons equality will not be axiomatized but handled via substitution rules.
We extend the sequent calculus {\LK} to the calculus {\LKeq} by allowing 
sequents of the form $\seq t = t$ as initial sequents and adding the following rules:

{\small
\[
\infer[\ElI]{A[t], \Gamma,\Pi \seq \Delta,\Lambda}{
  \Gamma \seq \Delta, s=t
  &
  A[s], \Pi \seq \Lambda
}
\qquad
\infer[\ElII]{A[t], \Gamma,\Pi \seq \Delta,\Lambda}{
  \Gamma \seq \Delta, t=s
  &
  A[s], \Pi \seq \Lambda
}
\]
\[
\infer[\ErI]{\Gamma,\Pi \seq A[t], \Delta,\Lambda}{
  \Gamma \seq \Delta, s=t
  &
  \Pi \seq A[s], \Lambda
}
\qquad
\infer[\ErII]{\Gamma,\Pi \seq A[t], \Delta,\Lambda}{
  \Gamma \seq \Delta, t=s
  &
  \Pi \seq A[s], \Lambda
}
\]
}
{\LKeq} is sound and complete for predicate logic with equality.

For convenience we write a substitution $[\sel{x_1}{t_1}, \ldots,
\sel{x_n}{t_n}]$ in the form $\unsubst{\bar{x}}{\bar{t}}$ for $\bar{x} =
(x_1,\ldots,x_n)$ and $\bar{t} = (t_1,\ldots,t_n)$.
A {\em strong quantifier} is a $\forall$ ($\exists$) quantifier with positive
(negative) polarity. We restrict our investigations to end-sequents in prenex
form without strong quantifiers.

\begin{definition}\label{def.sigmaIseq}
A {\em $\Sigma_1$-sequent} is a sequent of the form
$$\forall x_1\cdots \forall x_{k_1}F_1, \ldots, \forall x_1\cdots \forall x_{k_p}F_p
\seq
\exists x_1\cdots \exists x_{k_{p+1}}F_{p+1}, \ldots, \exists x_1\cdots \exists x_{k_q}F_q.$$
for quantifier free $F_i$.
\end{definition}

Note that the restriction to $\Sigma_1$-sequents does not constitute a
substantial restriction as one can transform every sequent into a
validity-equivalent $\Sigma_1$-sequent by skolemisation and prenexing.

\begin{definition}
A sequent $S$ is called {\em E-valid} if it is valid in predicate logic with
equality; $S$ is called a {\em quasi-tautology} if $S$ is quantifier-free and
E-valid.
\end{definition}

\begin{definition}
The length of a proof $\varphi$, denoted by $\com{\varphi}$, is defined as the
number of inferences in $\varphi$. The quantifier-complexity of $\varphi$,
written as $\comq{\varphi}$, is the number of weak quantifier-block
introductions in $\varphi$.
\end{definition}

\subsection{Extraction of terms}

Herbrand sequents of a sequent $S$ are sequents consisting of instantiations of
$S$ which are quasi-tautologies. The formal definition is:
\begin{definition}\label{def.Hseq}
Let $S$ be a $\Sigma_1$-sequent as in Definition~\ref{def.sigmaIseq} and let
$H_i$ be a finite set of $k_i$-vectors of terms
for every $i \in \{1,\ldots,q\}$. We define 
$\Fcal_i = \{F_i\unsubst{\bar{x_i}}{\bar{t}} \mid \bar{t}\in  H_i\}$ if $k_i > 0$
and $\Fcal_i = \{ F_i \}$ if $k_i = 0$. Let
$$S^*:\;\;\Fcal_1 \union \cdots  \union \Fcal_p  \seq \Fcal_{p+1} \union \ldots \union \Fcal_q.$$
If $S^*$ is a quasi-tautology then it is called a {\em Herbrand sequent} of $S$
and $H\colon (H_1,\ldots,H_q)$ is called a {\em Herbrand structure} of $S$. We
define the {\em size} of $S^*$ as $\sum^q_{i=1}|H_i|$.
\end{definition}
Note that, in the size of a Herbrand sequent, we only count the formulas
obtained by instantiation.

\begin{example}\label{ex.Hseq}
Consider the language containing a constant symbol $a$, unary function symbols
$f,s$,  a binary predicate symbol $P$, and the sequent $S$ defined below. We
write $f^n,s^n$ for $n$-fold iterations of $f$ and $s$ and omit parentheses
around the argument of a unary symbol when convenient. Let
\[S: P(f^4a,a), \forall x.fx = s^2x, \forall xy (P(sx,y) \impl P(x,sy)) \seq P(a,f^4a)\]
and $H = (H_1,H_2,H_3,H_4)$ for
{\small
\[
\begin{array}{l}
H_1 = \emptyset,\ H_4 = \emptyset,\ H_2 = \{a,fa,f^2a,f^3a\},\\
H_3 = \{(s^3f^2a,a), (s^2f^2a,sa),  (sf^2a,s^2a), (f^2a,s^3a), (s^3a,f^2a), (s^2a,sf^2a), (sa,s^2f^2a), (a,s^3f^2a)\}.
\end{array}
\]
}
Then 
{\small
\begin{eqnarray*}
\Fcal_1 &=& \{P(f^4a,a)\},\ \Fcal_4 = \{P(a,f^4a)\},\ \Fcal_2 = \{fa = s^2a,\ f^2a=s^2fa,\ f^3a = s^2f^2a,\ f^4a = s^2f^3a\}\\
\Fcal_3 &=& \{P(s^4f^2a,a) \impl P(s^3f^2a,sa), P(s^3f^2a,sa) \impl P(s^2f^2a,s^2a), P(s^2f^2a,s^2a) \impl P(sf^2a,s^3a),\\
        & & P(sf^2a,s^3a) \impl P(f^2a,s^4a), P(s^4a,f^2a) \impl P(s^3a,sf^2a), P(s^3a,sf^2a) \impl P(s^2a,s^2f^2a),\\
        & & P(s^2a,s^2f^2a) \impl P(sa,s^3f^2a), P(sa,s^3f^2a) \impl P(a,s^4f^2a)\}.
\end{eqnarray*}
}
A Herbrand-sequent $S^*$ corresponding to $H$ is then $\Fcal_1 \union \Fcal_2 \union \Fcal_3 \seq \Fcal_4$. Note that $ fa = s^2a,\ f^2a=s^2fa,\ f^3a = s^2f^2a,\ f^4a
= s^2f^3a \models f^4a = s^8a$.\\
The size of $S^*$ is 12. $S^*$ is a quasi-tautology but not a tautology.
\end{example}

%
\begin{theorem}[mid-sequent theorem]\label{thm.midseqthm}
Let $S$ be a $\Sigma_1$-sequent and $\pi$ a cut-free proof of $S$.
Then there is a Herbrand-sequent $S^*$ of $S$
s.t.\ $\com{S^*}\leq \comq{\pi}$.
\end{theorem}
%
%
\begin{proof}
This result is proven in~\cite{Gentzen34Untersuchungen} for {\LK}, but the proof
for {\LKeq} is basically the same. By permuting the
inference rules, one obtains a proof $\pi'$ from $\pi$ which has an upper part
containing only propositional inferences and the equality rules (which can be
shifted upwards until they are applied to atoms only) and a lower part
containing only quantifier inferences. The sequent between these parts is called
{\em mid-sequent} and has the desired properties.
\end{proof}

$S^*$ can be obtained by tracing the introduction of quantifier-blocks in
the proof, which for every formula $Q \bar{x}_i. F_i$ in the sequent (where
$Q\in\{\forall,\exists\}$) yields a
set of term tuples $H_i$, and then computing the sets of formulas $\Fcal_i$.

The algorithm for introducing cuts described here relies on computing a
compressed representation of the Herbrand structure, which is explained in
Section \ref{sec.grammar}. Note, though, that the Herbrand structure $(H_1, \ldots,
H_q)$ is a list of sets of term tuples (i.e. each $H_i$ is a set of tuples
$\overline{t}$ used to instantiate the formula $F_i$). In order to facilitate
computation and representation, we will add to the language fresh function
symbols $f_1, \ldots, f_q$. Each $f_i$ will be applied to the tuples of the set
$H_i$, therefore transforming a list of sets of tuples into a set of terms. In
this new set, each term will have an $f_k$ as its head symbol, that indicates to
which formula the arguments of $f_k$ belong.

\begin{example}
\label{ex.termset}

Using this new notation, the Herbrand structure $H$ of the previous example is
now represented as the set of terms:
$$T: \{ f_2(a), f_2(fa), f_2(f^2a), f_2(f^3a), f_3(s^3f^2a, a), f_3(s^2f^2a, sa), \ldots, f_3(a, s^3f^2a)\}.$$
\end{example}

Henceforth we will refer to the transformed Herbrand structure as the \emph{term set} of a proof.

\section{Computing a Decomposition}
\label{sec.grammar}

We shall now describe an algorithm for computing a compressed representation
of a term set $T$.
Term sets will be represented by decompositions which are defined as follows:

\begin{definition}
Let $T = \{t_1,\ldots,t_n\}$ be a set of ground terms. A decomposition $D$ of $T$ is a pair, written as
$U \circ_{\bar{\alpha}} W$, where $U$ is a set of terms containing the variables $\alpha_1,\ldots,\alpha_m$,
and $W = \left\{\bar{w}_1 = \colvec{w_{1,1}\\ \vdots\\w_{1,m}},\ldots,\bar{w}_q = \colvec{w_{q,1}\\ \vdots\\ w_{q,m} } \right\}$
is a set of vectors of ground terms s.t.
$T = U \circ_{\bar{\alpha}} W = \{ u[\bar{\alpha} \backslash \bar{w}]\ |\ u \in U, \bar{w} \in W\}$.
The size of a decomposition $U \circ_{\bar{\alpha}} W$ is $|U| + |W|$.
When it is clear that the variables in question are $\alpha_1,\ldots,\alpha_m$, we just write $U \circ W$.
\end{definition}

In~\cite{HetzlXXAlgorithmic} we have given an algorithm that
treats the special case where $m=1$. 
Here, we will extend that approach with a {\em generalized $\Delta$-vector}
\DeltaG, which, together with a so-called $\Delta$-table, can compute
decompositions with an arbitrary $m$.
\DeltaG, given in Algorithm~\ref{alg.DeltaG}, computes a {\em simple decomposition}, i.e.\ a
decomposition with only one term in $U$. The $\Delta$-table stores such decompositions
and builds more complex ones out of them. Due to space reasons, the algorithm can
only be sketched here, for details the interested reader is referred to the 
technical report~\cite{deltavec.report}.

\begin{definition}
  Let $T$ be a term set. The $\Delta$-table for $T$ is a list of key/value-entries, where each
  entry is of the form $W \Rightarrow (U_\Delta = \{ (u_1,T_1),\ldots,(u_n,T_n) \})$,
  where $W$ is a list of ground term vectors, $u_i$ is a term containing variables,
  and $T_i$ is a subset of $T$ s.t. the following two conditions are satisfied:

  \begin{enumerate}
    \item For every entry $W \Rightarrow \{ (u_1,T_1),\ldots,(u_n,T_n) \}$, $\{u_i\} \circ W$ is a
          decomposition of $T_i$ (for $1 \leq i \leq n$).
    \item For every $T' \subseteq T$, there is a pair $W \Rightarrow U_\Delta$ in the $\Delta$-table s.t.
          $(u,T') \in U_\Delta$.
  \end{enumerate}
\end{definition}
\begin{algorithm}
\caption{Generalized $\Delta$-vector \DeltaG}
\label{alg.DeltaG}
\begin{algorithmic}
\Function{\DeltaG}{$t_1,\ldots,t_n$: a list of terms}
  \State \Return transposeW(\DeltaG'($t_1,\ldots,t_n$))
\EndFunction
\Function{\DeltaG'}{$t_1,\ldots,t_n$: a list of terms}
  \If{$t_1 = t_2 = \ldots = t_n \land n > 0$}
    \Comment case 1: all terms identical
    \State \Return $(t_1,())$
  \ElsIf{$t_i = f(t_1^i,\ldots,t_m^i)$ for $1\leq i\leq n$}
    \Comment case 2: recurse
    \State $(\bar{w}_1,\dots,\bar{w}_q) \gets \bigsqcup\limits_{1 \leq j \leq m} \pi_2(\Delta_G(t_j^1,\ldots,t_j^n))$
    \Comment $\bigsqcup \equiv$ concatenation
    \State $u_j \gets \pi_1(\Delta_G(t_j^1,\ldots,t_j^n))$ for all $j \in \{1,\ldots,m \}$
    \State \Return merge($f(u_1,\ldots,u_m), (\bar{w}_1,\ldots,\bar{w}_q)$)
    \Comment merge all $\alpha_i$, $\alpha_j$ where $\bar{w}_i = \bar{w}_j$
  \Else
    \Comment case 3: introduce new $\alpha$
    \State \Return $(\alpha_\mt{FRESH}, (t_1,\ldots,t_n))$
  \EndIf
\EndFunction
\end{algorithmic}
\end{algorithm}
We build the $\Delta$-table as follows: for every $T' \subseteq T$, we compute
$\Delta_G(T') = (u, W)$ and insert $(u,T')$ with the key $W$ (if an entry
$W \Rightarrow U_\Delta$ already exists, we replace it with $W \Rightarrow U_\Delta \cup \{(u,T')\}$).
We then iterate over the $\Delta$-table and, for each entry $W \Rightarrow U_\Delta$,
try to find a subset $\{(u_{i_1},T_{i_1}),\ldots,(u_{i_n},T_{i_n})\}$ of $U_\Delta$ s.t.
$\{u_{i_1},\ldots,u_{i_n}\} \circ W$ is a decomposition of $T$. This is called {\em folding} the $\Delta$-table.

\begin{theorem}[Soundness and completeness]
\label{theo.decompositionCorrectness}
Let $T$ be a term set. If $U \circ W$ is extracted from folding the $\Delta$-table,
then $U \circ W$ is a decomposition of $T$. 
Conversely, if there exists a decomposition $U \circ W$ of $T$ of size $n$, folding the $\Delta$-table
will return at least one decomposition of size $n' \leq n$.
\end{theorem}

\begin{proof}See Appendix.\end{proof}

In fact, a stronger result holds: for every decomposition, there exists a unique normal form,
and folding the $\Delta$-table will only return decompositions in such normal form. For details,
see~\cite{deltavec.report}. To illustrate the algorithm, we compute a decomposition of the term
set of Example~\ref{ex.termset}. We remark that, for our cut-introduction method, we are interested
in a {\em decomposition $(U_1,\ldots,U_q)\circ W$ of a Herbrand structure} $H=(H_1,\ldots,H_q)$
which has the property that
$H_j=\{u[\bar{\alpha} \backslash \bar{w}] \mid u\in U_j, \bar{w} \in W\}$.
This is trivially obtained from a decomposition $U\circ W$ of the term set of a Herbrand structure
by setting $U_j=\{u\mid f_j(u)\in U\}$. 

\begin{example}\label{ex.DeltaG}

Let $T = T_2 \cup T_3$ with
{\small
$$
\begin{array}{l}
T_2 = \{ t_1 = f_2(a), t_2 = f_2(fa), t_3 = f_2(f^2a), t_4 = f_2(f^3a)\}\\
T_3 = \{t_5 = f_3(s^3f^2a,a), t_6 = f_3(s^2f^2a,sa),  \ldots, t_{12} = f_3(a,s^3f^2a) \}
\end{array}
$$ 
}
be a term set corresponding to the Herbrand structure $H = (H_1,H_2,H_3,H_4)$:
{\small
$$
\begin{array}{l}
H_1 = \emptyset$, $H_4 = \emptyset, H_2 = \{a,fa,f^2a,f^3a\},\\
H_3 = \{(s^3f^2a,a), (s^2f^2a,sa), \ldots, (a,s^3f^2a)\}
\end{array}$$
}

We now compute \DeltaG\ for every subset of $T$ --- consider for instance the subset $T' = \{ f_3(s^3f^2a,a), f_3(s^2f^2a,sa) \} \subseteq T$:
{\small
$$
  \Delta_G(f_3(s^3f^2a,a), f_3(s^2f^2a,sa)) = (f_3(s^2\alpha_1, \alpha_2), \left\{\colvec{sf^2a\\a}, \colvec{f^2a\\sa}\right\}) = (u,W).
$$
}
If the $\Delta$-table already has an entry $W \Rightarrow U_\Delta$, we add $(u,T')$
to $U_\Delta$. If not, we insert a new entry $W \Rightarrow \{(u,T')\}$.
After \DeltaG\ has been computed for all subsets, we iterate through it, looking
for simple decompositions that can be composed into a decomposition of $T$. We find the entry

{\small
$$
\begin{array}{l l l}
W &\Rightarrow& U_\Delta^1 \cup U_\Delta^2\\
U_\Delta^1 &=& \{ (f_2(\alpha_1),\{t_1,t_3\}),\ (f_2(f\alpha_1), \{t_2,t_4\}),\ (f_2(\alpha_2), \{t_1,t_3\})\ (f_2(f\alpha_2), \{t_2,t_4\})\}\\
U_\Delta^2 &=& \{ (f_3(s^3\alpha_1,\alpha_2), \{t_5, t_9\}),\ (f_3(s^2\alpha_1,s\alpha_2), \{t_6, t_{10}\}),\\
& & \ \ (f_3(s\alpha_1,s^2\alpha_2), \{t_7, t_{11}\}),\ (f_3(\alpha_1,s^3\alpha_2), \{t_8, t_{12}\})\}\\
W &=& \left\{\colvec{f^2a\\a}, \colvec{a\\f^2a}\right\}.
\end{array}
$$
}
and can see that $U_\Delta^1 \circ W = T_2$ and $U_\Delta^2 \circ W = T_3$
Therefore, $(U_\Delta^1 \cup U_\Delta^2) \circ W$ is a decomposition of $T$.
We then translate this decompositions $T$ back into a decomposition of $H$ by removing the
function symbols $f_2$ and $f_3$ from $U_\Delta^1$ \& $U_\Delta^2$:
{\small
\begin{eqnarray*}
U &=& (U_1,U_2),\\
U_1 &=& \{\alpha_1,\ f\alpha_1,\ \alpha_2,\ f\alpha_2\},\\
U_2 &=& \{(s^3\alpha_1,\alpha_2),\ (s^2\alpha_1,s\alpha_2),\ (s\alpha_1,s^2\alpha_2),\ (\alpha_1,s^3\alpha_2)\},\\
W &=& \left\{\colvec{f^2a\\a}, \colvec{a\\f^2a}\right\}.
\end{eqnarray*}
}
\end{example}
%
%
%
%
%
%
%
%

\section{Computing a Cut-Formula}

After having computed a decomposition as described in Section~\ref{sec.grammar},
the next step consists in computing a cut-formula based on that decomposition.
A decomposition $D$ specifies the instances of quantifier blocks in a
proof with a $\forall$-cut, but does not contain information about the
propositional structure of the cut formula to be constructed. The problem to
find the appropriate propositional structure is reflected in the following
definition.

\begin{definition}\label{def.sextHseq}
Let $S$ be a $\Sigma_1$-sequent and $F_i, k_i$ as in
Definition~\ref{def.sigmaIseq}, $H$ be a Herbrand structure for $S$, 
and $D\colon U \circ W$
a decomposition of $H$ with $V(D) = \{\alpha_1,\ldots,\alpha_n\}$. Let $U = (U_1, \ldots , U_q)$  and $W = \{\bar{w}_1, \ldots, \bar{w}_k\}$, where the $\bar{w}_j$
are $n$-vectors of terms not containing variables in $V(D)$, and $\Fcal'_i = \{F_i\unsubst{\bar{x}_i}{\bar{t}} \mid \bar{t} \in U_i\}$
for $k_i > 0$ and $\Fcal'_i = \{ F_i \}$ for $k_i = 0$. Furthermore let $X$ be an $n$-place
predicate variable. Then the sequent
$$S^\sim:\;\; X\bar{\alpha} \impl \Land\nolimits^k_{i=1}X \bar{w}_i, \Fcal'_1, \ldots , \Fcal'_p \seq \Fcal'_{p+1}, \ldots, \Fcal'_q.$$
is called a {\em schematic extended Herbrand sequent} of $S$ w.r.t. $D$. The
{\em size} of $S^\sim$, denoted by $\com{S^\sim}$, is defined as $k+
\sum^q_{i=1}|U_i|$.
\end{definition}

\begin{definition}\label{def.sextHsol}
Let $S^\sim$ be a schematic extended Herbrand sequent of $S$ w.r.t.~a
decomposition $D$ as in Definition~\ref{def.sextHseq} and $A$ be a formula with
$V(A) \IN \{\alpha_1,\ldots,\alpha_n\}$. Then the second-order substitution
$\sigma\colon \unsubst{X}{\lambda \bar{\alpha}.A}$ is a solution of $S^\sim$ if
$S^\sim\sigma$ is a quasi-tautology; in this case $S^\sim\sigma$ is called an
{\em extended Herbrand sequent}. The size of $S^\sim\sigma$ is defined as
$\com{S^\sim}$.
\end{definition}
Theorem \ref{theo.extHseq} in Section \ref{sec.proof_build} shows that, from a
solution of a schematic extended
Herbrand sequent $S^\sim$ of $S$, we can define a proof $\psi$ of $S$ with a
$\forall$-cut and $\comq{\psi} = \com{S^\sim}$.
The question remains whether every schematic extended Herbrand sequent is
solvable. We show below that this is indeed the case.

Let $S^\sim$ as in Definition~\ref{def.sextHseq}. We define
\[F[l] = \Land \Union\nolimits^p_{i=1} \Fcal'_i  \ \mbox{and}\ F[r] = \Lor\Union\nolimits^q_{i=p+1} \Fcal'_i.\]
\begin{definition}\label{def.canonsubs}
Let $S^\sim$ be a schematic extended Herbrand sequent of $S$ as in Definition~\ref{def.sextHseq}. We define the {\em canonical formula} $C(S^\sim)$ of $S^\sim$ as $F[l]
\land \neg F[r]$. The substitution $\unsubst{X}{\lambda \bar{\alpha}.C(S^\sim)}$ is called the {\em canonical substitution} of $(S,S^\sim)$.
\end{definition}
\begin{theorem}\label{theo.canonsol}
Let $S$ be a $\Sigma_1$-sequent, and $S^\sim$ be a schematic extended Herbrand
sequent of $S$. Then the canonical substitution is a solution of $S^\sim$.
\end{theorem}
\begin{proof}
Let $S^\sim$ be  a schematic extended Herbrand sequent as in Definition~\ref{def.sextHseq} and $C(S^\sim)$ be the canonical formula of $S^\sim$. We have to prove that
$$S_1:\;\; C(S^\sim)(\bar{\alpha}) \impl \Land\nolimits^k_{i=1} C(S^\sim)(\bar{w}_i),\  F[l] \seq F[r]$$
is a quasi-tautology. But, by definition of $C(S^\sim)$, $S_1$ is equivalent to
$$S_2:\;\; (F[l] \land \neg F[r]) \impl \Land\nolimits^k_{i=1}(F[l] \land \neg F[r])(\bar{w}_i), (F[l] \land \neg F[r]) \seq.$$
Clearly $S_2$ is a quasi-tautology if the sequent $S_3$, defined as
$$S_3:\;\; \Land\nolimits^k_{i=1}(F[l] \land \neg F[r])(\bar{w}_i)\seq$$
is a quasi-tautology. But, by $D = U \circ W$ being a decomposition of $H$, $S_3$ is logically equivalent to the Herbrand sequent $S^*$ defined over $H$, which (by definition) is
a quasi-tautology.
\end{proof}

\begin{example}\label{ex.sextHseq}
Let
$$S:\;\; P(f^4a,a), \forall x.fx = s^2x, \forall xy (P(sx,y) \impl P(x,sy)) \seq P(a,f^4a)$$
like in Example~\ref{ex.Hseq} and $D$ be the decomposition $U \circ W$ of $H$
constructed in Example~\ref{ex.DeltaG}. We have
{\small
\begin{eqnarray*}
U &=& (U_1,U_2),\\
U_1 &=& \{\alpha_1,\ f\alpha_1,\ \alpha_2,\ f\alpha_2\},\\
U_2 &=& \{(s^3\alpha_1,\alpha_2),\ (s^2\alpha_1,s\alpha_2),\ (s\alpha_1,s^2\alpha_2),\ (\alpha_1,s^3\alpha_2)\},   \\
W &=& \left\{\colvec{f^2a\\a}, \colvec{a\\f^2a}\right\}.
\end{eqnarray*}
}
The corresponding schematic extended Herbrand sequent $S^\sim$ is
\[
\begin{array}{l}
X(\alpha_1,\alpha_2) \impl (X(f^2a,a) \land X(a,f^2a)),\\
f\alpha_1 =s^2\alpha_1,\ f^2\alpha_1 = s^2f\alpha_1,\ f\alpha_2 = s^2\alpha_2,\ f^2\alpha_2 = s^2f\alpha_2,\\
P(s^4\alpha_1,\alpha_2) \impl P(s^3\alpha_1,s\alpha_2),\ P(s^3\alpha_1,s\alpha_2) \impl P(s^2\alpha_1,s^2\alpha_2), P(s^2\alpha_1,s^2\alpha_2) \impl P(s\alpha_1,s^3\alpha_2),\\
P(s\alpha_1,s^3\alpha_2) \impl P(\alpha_1,s^4\alpha_2), P(f^4a,a) \seq P(a,f^4a).
\end{array}
\]
Its canonical formula $C(S^\sim)$ which we write as $A(\alpha_1,\alpha_2)$ is
\[
\begin{array}{l}
\Land^2_{i=1}(f\alpha_i =s^2\alpha_i \land\ f^2\alpha_i = s^2f\alpha_i)  \land \\
\Land^3_{i=0}(P(s^{4-i}\alpha_1,s^{i}\alpha_2) \impl P(s^{4-i-1}\alpha_1,s^{i+1}\alpha_2)) \land  P(f^4a,a) \land \neg P(a,f^4a).
\end{array}
\]
The canonical solution is $\unsubst{X}{\lambda \alpha_1
\alpha_2.A(\alpha_1,\alpha_2)}$ and the corresponding Herbrand sequent $S'$ is
like $S^\sim$ with $X(\alpha_1,\alpha_2) \impl (X(f^2a,a) \land X(a,f^2a))$
replaced by $A(\alpha_1,\alpha_2) \impl (A(f^2a,a) \land A(a,f^2a))$.
Note that $\com{S'} = 10$, while $\com{S^*} = 12$. So we obtained a compression of quantifier complexity.

\end{example}

\subsection{Improving the solution}

In the last section, we have shown that, given a decomposition $D$
of the termset of a cut-free proof of a $\Sigma_1$-sequent $S$, there exists a canonical solution
to the schematic extended Herbrand sequent induced by $S,D$, which gives rise to a proof with
a $\forall$-cut. The canonical solution
need not be the best solution for a given purpose; indeed it is often not symbol-minimal, for
example. Hence this section is devoted to describing
an algorithm for finding better solutions. We will consider 
E-validity of quantifier-free formulas $F$
containing free variables; by ``$F$ is E-valid'' we mean to say ``the universal closure of $F$ is E-valid''.
Throughout this section, we consider 
a fixed $\Sigma_1$-sequent $S$ using the notation of Definition~\ref{def.sigmaIseq},
a fixed decomposition $D=(U_1,\ldots,U_q)\circ W$, with 
$W=\{\bar{w}_i\mid 1\leq i \leq k\}$, of a Herbrand structure $H$
of $S$, along with the schematic extended Herbrand sequent 
%
$S^\sim$ induced by $S,H,D$, using the notation of Definition~\ref{def.sextHseq}. 
We will abbreviate 
$\Fcal'_1 \union \cdots \union \Fcal'_p$ by $\Gamma$ and $\Fcal'_{p+1}\union \cdots\union \Fcal'_q$ by $\Delta$,
and write ``$A$ is a solution'' for ``$\unsubst{X}{\lambda \bar{x}.A}$ is a solution for $S^\sim$'' (note that
we will consider the names $\bar{x}$ fixed). In this section, we will focus our attention on solutions in conjunctive
normal form (CNF), which always exist since the solution property is semantic 
(if $A$ is a solution and $A\Leftrightarrow B$
is E-valid, then $B$ is a solution). A clause $C$ is said to be $\bar{x}$-free if it contains no symbol from $\bar{x}$.

%

The algorithm we will present will involve generating E-consequences of formulas. 
Although in principle an abstract analysis of our algorithm based on a notion of 
{\em E-consequence generator} can be performed, we have chosen, for lack of space,
to present only the concrete E-consequence generator used in our implementation.

We now present this E-consequence generator, which is based on {\em forgetful reasoning}.
Let $C_1,C_2$ be two clauses, then denote
the set of propositional resolvents of $C_1, C_2$ by $\res(C_1, C_2)$ and the set of clauses that can
be obtained from $C_1,C_2$ by ground paramodulation by $\para(C_1, C_2)$.
Letting $F$ be a formula with CNF $\{C_i\}_{i\in I}$ we define
\[
  \forgetOp(F) = \{ C\wedge\bigwedge\nolimits_{i\in I\setminus\{j,k\}} C_i \mid C\in\res(C_j,C_k)\cup\para(C_j,C_k)\}.
\]

Using $\forgetOp$, we can now present Algorithm~\ref{alg:simp_sol}: the solution-finding algorithm $\SF{\forgetOp}$.
It prunes a solution $A$ of $\bar{x}$-free clauses, then recurses upon those consequences of the pruned 
$A$ generated by $\forgetOp$
which pass a certain E-validity check, finally returning a set of formulas (which will all be solutions).
\begin{algorithm}
\caption{\SF{\forgetOp}}
\label{alg:simp_sol}
\begin{algorithmic}
\Function{\SF{\forgetOp}}{$A$: solution in CNF}
\State $A \gets A$ without $\bar{x}$-free clauses 
  \State $S \gets \{A\}$
  \For{$B \in \forgetOp(A)$}
  \If{$B\unsubst{\bar{x}}{\bar{w}_1},\ldots,B\unsubst{\bar{x}}{\bar{w}_k},\Gamma\seq\Delta$ is E-valid}
      \Comment $B$ is a solution
      \State $S \gets S\cup \SF{\forgetOp}(B)$
    \EndIf
  \EndFor
  \State \Return $S$
\EndFunction
\end{algorithmic}
\end{algorithm}
We have the following result, which is derived essentially from the algebraic structure
of the solution space which is sketched in the following section.
\begin{theorem}[Soundness \& Termination]\label{thm:sf_sound_complete}
Let $A$ be any solution in CNF.
Then $\SF{\forgetOp}$ terminates on $A$ and, for all $B\in\SF{\forgetOp}(A)$,
$B$ is a solution.
\end{theorem}

\begin{example}
  Consider the canonical formula $C(S^\sim)$ of Example~\ref{ex.sextHseq}.
  Then $\SF{\forgetOp}$ generates the CNF
  \[
    F(\alpha_1,\alpha_2):\; f^2\alpha_1 = s^4\alpha_1 \land f^2\alpha_2=s^4\alpha_2 \land (\neg P(s^4\alpha_1,\alpha_2) \lor P(\alpha_1,s^4\alpha_2))
  \]
  for the CNF of $C(S^\sim)$ by applying paramodulation twice to equational atoms and resolution
  thrice to the clauses corresponding to the implications between the $P$-atoms. It can be
  checked that $\lambda\alpha_1\alpha_2.F(\alpha_1,\alpha_2)$ is a solution for $S^\sim$
  which is smaller than the canonical solution.
\end{example}
\subsection{The solution space}
This section is dedicated to describing the space of solutions. 
The following result summarizes the algebraic properties of
the solution space which are exploited in $\SF{\forgetOp}$;
in particular these properties allow one to prove
the correctness of $\SF{\forgetOp}$ 
in the sense of Theorem~\ref{thm:sf_sound_complete}.
By $\formulas$ we denote the set of propositional formulas built from the
atoms of the canonical solution of $S^{\sim}$.
\begin{theorem}
Define 
the equivalence relation $\sim$ on $\formulas$ by setting
$F\sim G$ iff $F\Leftrightarrow G$ is E-valid (i.e.~$\formulas/{\sim}$ is
the Lindenbaum-Tarski algebra of $\formulas$ w.r.t.~the theory of equality).
Let $\solutions\subseteq\formulas$ 
be the set of solutions in that signature.
Then $\mathcal{B}=(\formulas/{\sim}, \land, \lor, \neg, \bot, \top)$ is a Boolean algebra,
$(\solutions/{\sim}, \land)$ is a convex subalgebra of the meet-semilattice reduct
$(\formulas/{\sim}, \land)$ of $\mathcal{B}$, and the canonical solution $C$
is the least element of $(\solutions/{\sim}, \land)$.
\end{theorem}
\tikzstyle{vertex}=[circle,minimum size=12pt,draw,inner sep=0pt]
\tikzstyle{subvertex}=[circle,minimum size=12pt,fill=gray,inner sep=0pt]
\tikzstyle{edge} = [draw,-]
\tikzstyle{subedge} = [draw=gray,-]
\begin{figure}
\begin{tikzpicture}[scale=0.3]
\foreach \pos/\name in {{(7,0)/xyzu}, {(1,4)/xyz},{(9,4)/xzu}, {(13,4)/yzu},
                        {(3,8)/xz}, {(6,8)/yz}, {(12,8)/yu},
                        {(15,8)/zu}, {(5,12)/y}, {(9,12)/z}, {(13,12)/u},
                        {(7,16)/empt}, {(1,12)/x}}
  \node[vertex] (\name) at \pos {};

\foreach \pos/\name in {{(5,4)/xyu}, {(0,8)/xy}, {(9,8)/xu}}
  \node[subvertex] (\name) at \pos {};

\foreach \source/\dest in {x/xy, x/xu, empt/x, empt/y, empt/z, empt/u,
                           x/xz, y/xy, y/yz, y/yu, z/xz,
                           z/yz, z/zu, u/xu, u/yu, u/zu, xy/xyz,
                           xz/xyz, xz/xzu, yz/xyz, yz/yzu,
                           xu/xzu, yu/xyu, yu/yzu, zu/xzu, 
                           zu/yzu, xyz/xyzu, xyu/xyzu, xzu/xyzu, yzu/xyzu}
  \path[edge] (\source) -- (\dest);

\foreach \source/\dest in {xy/xyu, xu/xyu}
\path[subedge] (\source) -- (\dest);

\node (text) at (-8, 0) {canonical solution $C$};

\path[draw,->,thick] (text) -- (xyu);
\end{tikzpicture}\caption{$\formulas/{\sim}$ with solution space $\solutions/{\sim}$ indicated}
\label{fig:sol_space}
\end{figure}
Figure~\ref{fig:sol_space} visualizes this result for a language
with $4$ atoms.
The algorithm $\SF{\forgetOp}$ can be seen as searching upwards through the solution
semilattice, starting at the bottom element $C$. For any (representative of an)
element $F$ of this semilattice, $\forgetOp(F)$ contains formulas which either correspond
to $F$, or an element above $F$. Whenever we reach an element $G$ that is not a solution,
we know that no other element above $G$ is a solution since the solution space
is convex.
%

As indicated above, the algorithm $\SF{\forgetOp}$ can in principle be stated using
E-consequence generators different from $\forgetOp$, and one can formulate
general properties of such generators from which important results, such as 
completeness, can be derived. Furthermore, it will be possible to apply methods
from automated theorem proving in presence of equality to the development of
practically useful E-consequence generators. We leave such further theoretical and empirical
investigation of ``good'' E-consequence generators for future work.
\subsection{Proof with cut}
\label{sec.proof_build}

\begin{theorem}\label{theo.extHseq}
Let $S^\sim$ be an extended Herbrand sequent of a $\Sigma_1$-sequent $S$. Then
$S$ has a proof with a $\forall$-cut s.t. $\comq{\varphi} = \com{S^\sim}$. 
\end{theorem}
\begin{proof}
As in~\cite{HetzlXXAlgorithmic}. Note that the quantifier-blocks in the cut and the
equality rules do not change the measured number of weak quantifier-block
introductions analyzed in the paper above. The main steps in the proof are the
following ones: let $S'$ be an extended Herbrand sequent obtained by the
solution $\unsubst{X}{\lambda \bar{\alpha}.A}$. Then a proof with cut formula
$\forall \bar{x}.A\unsubst{\bar{\alpha}}{\bar{x}}$ can be constructed where the
quantifier substitution blocks for the cut formula on the right-hand-side are
$\unsubst{\bar{x}}{\bar{w}}$ for $\bar{w} \in W$ while the cut formula on the
left-hand-side gets the substitution $\unsubst{\bar{x}}{\bar{\alpha}}$.  The
substitutions $\unsubst{\bar{x_i}}{\bar{t}}$ for $\bar{t} \in U_i$ are inserted
to introduce the quantifiers of the formula $F_i$ in the end-sequent.
\end{proof}

\begin{example}

Let $\Gamma = P(f^4a,a), \forall x.fx = s^2x, \forall xy (P(sx,y) \impl
P(x,sy))$ be the left-hand-side of $S$. Then, to the canonical solution corresponds
an {\LK}-proof $\psi$ of the form
{\small
\[
\infer[\cut+\contrstar]{\Gamma \seq P(a,f^4a)}
 { \infer[\forallr^*]{\Gamma \seq P(a,f^4a), \forall xy.A(x,y)}
   { \deduce{\Gamma \seq P(a,f^4a), A(\alpha_1,\alpha_2)}{(\psi_1)}
   }
   &
   \infer[\foralll^*]{\Gamma,\forall xy.A(x,y) \seq P(a,f^4a)}
    { \deduce{\Gamma,A(f^2a,a),A(a,f^2a) \seq P(a,f^4a)}{(\psi_2)}
    }
  }  
\]
}
where $\psi_1$ and $\psi_2$ are cut-free and $\psi_2$ contains only structural
and propositional inferences (in $\psi_2$ only $P(f^4a,a)$ is needed from
$\Gamma$). The quantifier inferences in $\psi_1$ use exactly the substitutions
encoded in $U_1$ and $U_2$. So we have $\comq{\psi} = 10$.

\end{example}

\section{Implementation and Experiments}

Summing up the previous sections, the structure of our cut-introduction
algorithm is the following:
\begin{algorithm}
\caption{Cut-Introduction}
\begin{algorithmic}
\Require {$\pi$: cut-free proof}
\State $T\gets \extractTS(\pi)$
\State $D \gets \getMinimalDecomposition(T)$
\State $C(\bar{x}) \gets \getCanonicalSolution(D)$
\State $F(\bar{x}) \gets \improveSolution(C(\bar{x}))$
\State \Return $\constructProof(F(\bar{x}))$
\end{algorithmic}
\end{algorithm}

Depending on whether the input proof $\pi$ contains equality reasoning or not
we either work modulo quasi-tautologies as described in this paper
or modulo tautologies (as described in~\cite{Hetzl12Towards,HetzlXXAlgorithmic})
in $\improveSolution$ and $\constructProof$. In $\getMinimalDecomposition$ we
can either compute decompositions with a single variable
as in~\cite{Hetzl12Towards,HetzlXXAlgorithmic} or with an unbounded number of variables
as described in Section~\ref{sec.grammar}. We denote these two variants with $\CI^1$
and $\CI^*$ respectively.

These algorithms have been implemented in the gapt-system\footnote{Generic Architecture
for Proof Transformations, \url{http://www.logic.at/gapt/}} which is a framework
for transforming and analyzing formal proofs. It is implemented in Scala and
contains data structures such as formulas, sequents,
resolution and sequent calculus proofs and algorithms like unification,
skolemization, cut-elimination as well as backends for several external solvers
and provers.
For deciding whether a quantifier-free formula is a
tautology we use MiniSat\footnote{\url{http://minisat.se/}}.
We use veriT\footnote{\url{http://www.verit-solver.org/}}
for deciding whether a quantifier-free formula is a quasi-tautology and
prover9\footnote{\url{http://www.cs.unm.edu/{\raise.17ex\hbox{$\scriptstyle\sim$}}mccune/prover9/}} for the actual 
proof construction based on the import described in~\cite{Hetzl13Understanding}.

We have conducted experiments on the prover9-part of the TSTP-library (Thousands
of Solutions of Theorem Provers, see~\cite{Sut10}). The choice of prover9 was motivated by the simple
and clean proof output format Ivy which makes proof import (comparatively) easy.
This library contains 6341 resolution proofs. Of those, 5254 can be parsed
and transformed into a sequent calculus proof using the transformation
described in~\cite{Hetzl13Understanding}. Of those, 2849
have non-trivial termsets (we call a term set trivial if every quantified formula
in the end-sequent is instantiated at most once).

The input data we have used for our experiments is this collection of proofs with
non-trivial term sets. In this collection 66\% 
use equality reasoning and hence must be treated with the method
introduced in this paper. The average term set size is 37,1 
but 46\% have a term set of size $\leq$ 10.
The experiments have
been conducted with version 1.6 of gapt on an Intel i5 QuadCore with 3,33GHz with
an allocation of 2GB heap space and a timeout of 60 seconds for the cut-introduction algorithm.

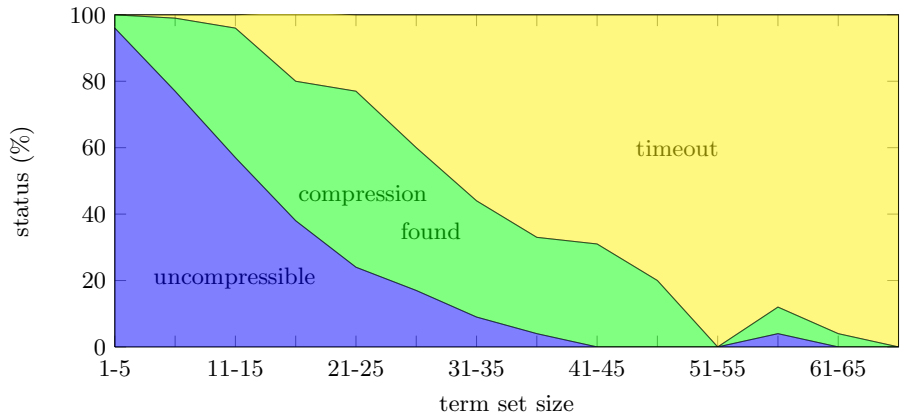
\begin{figure}[t]
\begin{center}
\begin{tikzpicture}
\begin{axis}[
width=12cm,
height=6cm,
ymin=0,
ymax=100,
stack plots=y,
enlargelimits=false,
xlabel=term set size,
ylabel=status ($\%$),
xtick={1,2,3,4,5,6,7,8,9,10,11,12,13,14},
xticklabels={1-5,,11-15,,21-25,,31-35,,41-45,,51-55,,61-65},
xticklabel style={
  anchor=base,
  yshift=-\baselineskip
}
]
\addplot[fill=blue,opacity=0.5] table[x=lineno,y=sep1] {r1866_status5.dat} \closedcycle;
\addplot[fill=green,opacity=0.5] table[x=lineno,y=sep2] {r1866_status5.dat} \closedcycle;
\addplot[fill=yellow,opacity=0.5] table[x=lineno,y=sep3] {r1866_status5.dat} \closedcycle;
\node at (5,15) [anchor=south west] {uncompressible};
\node at (29,40) [anchor=south west] {compression};
\node at (46,30) [anchor=south west] {found};
\node at (85,55) [anchor=south west] {timeout};
\end{axis} 
\end{tikzpicture}
\end{center}
\caption{$\CI^*$: return status by term set size}
\label{fig_status}
\end{figure}

On 19\% of the input proofs our algorithm terminates with finding a
compression, i.e.\ a non-trivial decomposition (of size at most that of the
original termset) and a proof with cut that realizes this decomposition.
On 49\% it terminates determining that the proof is uncompressible, 
more precisely: that there is no proof with a single $\forall$-cut which (by
cut-elimination) reduces to the given input term set and is
of smaller quantifier complexity, see~\cite{HetzlXXAlgorithmic}.
Figure~\ref{fig_status} depicts the return status (in percent) depending on
the size of the term set.
When reading this figure one should keep in mind the
relatively high number of small proofs (see above).
One can observe that proofs with term sets up
to a size of around 50 can be treated well by our current implementation, beyond
that the percentage of timeouts is very large. Small proofs -- unsurprisingly --
tend to be uncompressible.


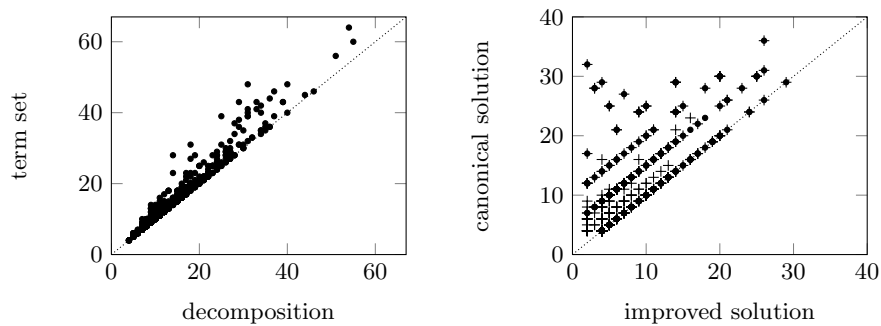
\begin{figure}
\begin{center}
\begin{tikzpicture}
\begin{axis}[
width=5.5cm,
xmin=0,
xmax=67,
ymin=0,
ymax=67,
enlargelimits=false,
xlabel=decomposition,
ylabel=term set]
\addplot[mark=*,only marks,mark size={1pt}] table[x=mgsize,y=tssize] {r1866_GCI_NoPara_ok.dat};
\addplot[smooth,densely dotted] plot coordinates {
  (0,0)
  (67,67)
};
\end{axis}
\end{tikzpicture}
\qquad
\begin{tikzpicture}
\begin{axis}[
width=5.5cm,
xmin=0,
xmax=40,
ymin=0,
ymax=40,
enlargelimits=false,
xlabel=improved solution,
ylabel=canonical solution]
\addplot[mark=*,only marks,mark size={1pt}] table[x=npminsolc,y=cansolc] {r1866_GCI_NoParaWithPara_okok.dat};
\addplot[mark=+,only marks,mark size={2pt}] table[x=wpminsolc,y=cansolc] {r1866_GCI_NoParaWithPara_okok.dat};
\addplot[smooth,densely dotted] plot coordinates {
  (0,0)
  (40,40)
};
\end{axis}
\end{tikzpicture}
\end{center}
\caption{Size Comparison}
\label{fig_size}
\end{figure}

In Figure~\ref{fig_size} we restrict our attention to runs terminating with
a compression. As one
can see from the diagram on the left, a significant reduction of quantifier-complexity
can be achieved by our method. The diagram on the right demonstrates that forgetful reasoning
is highly useful for improving the canonical solution. The points plotted as $\bullet$
are the result after using forgetful resolution only, the points plotted as $+$
are the result after forgetful resolution and paramodulation.

Our experiments also show that the generalization to the introduction of a block
of quantifiers introduced in this paper has a strong effect: of
the 548 proofs on which $\CI^*$ finds a compression, 22\% are found to
be uncompressible by $\CI^1$.




\section{Conclusion}

We have introduced a cut-introduction method that works modulo equality and
is capable of generating cut-formulas containing a block of quantifiers.
We have implemented our new method and have conducted a large-scale empirical evaluation
which demonstrates its feasibility on realistic examples.
Lessons learned from these experiments include that blocks of quantifiers
allow for significantly more proofs to be compressed and that forgetful reasoning
methods, while rough in theory, are highly useful for our application in practice.

As future work we plan to extend our method to work modulo (suitably specified)
equational theories. We also plan to evaluate our method on proofs produced by
Tableaux-provers and SMT-solvers. Another important, and non-trivial, extension will be to cope with cuts
that contain quantifier-alternations.

{\bf Acknowledgements.} The authors would like to thank Pascal Fontaine for help
with the veriT-solver and Geoff Sutcliffe for providing the prover9-TSTP
test set.

\bibliography{references}
\bibliographystyle{splncs03}

\end{document}